\documentclass[12pt]{article}
\usepackage{times}

\usepackage{amsthm,amsfonts,amssymb,amsmath}

\usepackage[T1]{fontenc}
\usepackage[cp1250]{inputenc}
\newcommand{\be}{\begin{equation}}
\newcommand{\ee}{\end{equation}}

\newcommand{\ba}{\begin{eqnarray}}
\newcommand{\ea}{\end{eqnarray}}

\usepackage{cite}
\usepackage{color}
\newtheorem{thm}{Theorem}[section]

\newtheorem{prop}[thm]{Proposition}

\newcommand{\p}{\partial}

\begin{document}
\centerline{\textbf{Decoupling the momentum constraints in general relativity}}

\date{\today}

\null

\centerline{ J Tafel} 
\noindent
\centerline{Institute of Theoretical Physics, University of Warsaw}
\centerline{ Ho\.{z}a 69, 00-681 Warsaw, Poland, tafel@fuw.edu.pl}

\begin{abstract}
We present a 2+1 decomposition of the vacuum initial conditions in general relativity.  For a constant mean curvature one of the momentum constraints decouples in quasi isotropic coordinates and it can be solved by quadrature.  The remaining  momentum constraints are written in the form of the tangential Cauchy-Riemann equation. Under additional assumptions its solutions  can be  written  in terms of integrals of known functions.   We show how to obtain  initial data with a marginally outer trapped  surface. A generalization of the Kerr data is presented.
\end{abstract}

\null


\null


\section{Introduction}

Vacuum initial data in general relativity  consist of  a Riemannian metric $\tilde g =\tilde g_{ij}dx^idx^j$ and a symmetric tensor $\tilde K=\tilde K_{ij}dx^idx^j$  given on a 3-dimensional manifold $S$. These data have to satisfy the constraint equations 
\begin{align}
\tilde\nabla_j\big(\tilde K^j_{\ i}-\tilde H\delta^j_{\ i}\big)&=0 \label{1}\\
  \tilde R+\tilde H^2-\tilde K_{ij}\tilde K^{ij}&=0,\label{2}
\end{align} 
where  $\tilde\nabla_i$ are covariant derivatives corresponding to $\tilde g$, $\tilde R$ is the Ricci scalar of $\tilde g$ and $\tilde H=\tilde K^i_{\ i}$. Tensors $\tilde g$ and  $\tilde K$ are interpreted, respectively, as the induced metric and the external curvature of $S$ embedded in a 4-dimensional spacetime developing from these data in accordance with the Einstein equations.

The conformal approach to the constraints  of Lichnerowicz, Choquet-Bruhat and York    (see \cite{cy,c,i} for a  review) is based on the following representation of initial data 
\begin{equation}\label{1a}
\tilde g_{ij}=\psi^4 g_{ij},\quad \tilde K_{ij}=\psi^{-2} K_{ij}+\frac 13\tilde H\psi^4 g_{ij}\ ,
\end{equation}
where $H=K^i_{\ i}=0$.
 The momentum constraint (\ref{1}) yields  equations for $g_{ij}$, $K_{ij}$, $\tilde H$ and $\psi$ with no derivatives of $\psi$, whereas the Hamiltonian constraint (\ref{2}) is equivalent to the Lichnerowicz equation
\begin{equation}\label{2a}
\bigtriangleup\psi=\frac 18 R\psi-\frac 18 K_{ij} K^{ij}\psi^{-7}+\frac {1}{12} \tilde H^2\psi^{5} ,
\end{equation}
where  $\bigtriangleup$ and $R$ are, respectively, the covariant Laplace operator and the Ricci scalar of  metric $g$. 

If $\tilde H=const$ (constant mean curvature  data)  the momentum constraint (\ref{1}) is equivalent  to 
\be
\nabla_jK^j_{\ i}=0 \ ,\ \ K^i_{\ i}=0\ .\label{3a}
\ee
 In this case one can first find a  solution $(g,K)$  of (\ref{3a}) and then consider equation  (\ref{2a}) for $\psi$.  Equations (\ref{3a}) were solved analytically or reduced to a simpler system only for  conformally flat initial metrics \cite{m,bl,by,bb} or  symmetric data \cite{bs,bp,d,cm,tj}.   Exact solutions of (\ref{2a}) are known  for  the simplest class of data with $K=0$ and conformally flat  $g$. In other cases, the best what one can do is to prove the existence of solutions   \cite{cb,d1,ma,ma1}  or to find them numerically (see \cite{al,b} for a review).
 
 The aim of this paper is to reduce equations (\ref{3a}) to a simpler system. We show that these equations decouple in  coordinates closely related to quasi isotropic coordinates of Smarr \cite{s}. One of the equations yields a component $W$ of $K$ as an  integral of free data. Then  the remaining two equations can be written as a single complex equation using the  Cauchy-Riemann structure related to the  metric $g$. Under additional assumptions solutions of this equation can be   also represented as  integrals of known functions.

In the last section we consider data admitting  a marginally outer trapped surface (MOTS)  which can be considered as an attribute of a black hole.  Known constructions  of such data are  based on the puncture method of Brill and Lindquist \cite{bl}, the conformal-imaging method of Misner \cite{m} or the boundary condition method proposed by Thornburg \cite{t}.   In the spirit of  Misner's approach  we define  a class of maximal non-conformally flat data with a reflection symmetry which assures existence of  MOTS. These data generalize the Kerr metric data but they don't have to be axially symmetric.

\section{The 2+1 decomposition of initial data}
We would like to find coordinates in which the momentum constraint (\ref{3a})  decouples and can  be  partly integrated. Let us foliate the initial manifold $S$ into surfaces given by constant levels of a function $\varphi$. 
In  coordinates $x^i=x^a,\varphi$, where $i=1,2,3$ and $a=1,2$,  an arbitrary   metric $g$ on $S$ can be written in the form
\begin{equation}
g=g_{ab}dx^adx^b+\alpha^2 (d\varphi+\beta_a dx^a)^2\ .\label{24}
\end{equation}
Coordinate transformations allow to impose up to 3 conditions on components of this metric. Most of them, including conditions satisfied in the Gauss coordinates, do not  simplify equation (\ref{3a}). 
The method of trial and error shows that probably the best coordinates in this respect  are  $x^a,\varphi$ such that the 2-dimensional  metric $g_{ab}$ is conformally flat
\be
g=\rho^2\delta_{ab}dx^adx^b+\alpha^2 (d\varphi+\beta_a dx^a)^2\ .\label{30}
\ee
In order to find these coordinates  one has  to solve  equation
\begin{equation}
 \xi_{,i}\xi^{,i}=0\label{30g}
\end{equation}
for  a complex function  $\xi$ 
such that $d\xi\wedge d\bar\xi\wedge d\varphi\neq 0$. Then $Re\xi$, $Im\xi$  and $\varphi$ are new coordinates in which metric takes the form  (\ref{30}).  Still  one can change coordinate  $\varphi$. This freedom allows to    reduce a number of functions in $g$ to three. For instance,  in this way one can obtain conditions satisfied by the quasi isotropic coordinates \cite{s}.

We are not able to describe all initial metrics which admit solutions of (\ref{30g}). 
Throughout the paper we  will  assume  existence of coordinates $x^a,\varphi$ such that $g$ is given by (\ref{30}). It is not very strong restriction compared  to the assumption $\tilde H=const$. We will work in a single coordinate system. Given a solution of (\ref{3a})  one can try to extend it to an acceptable initial manifold.  Coordinates $x^i$ may be analogs of the Cartesian coordinates in flat space but they may be also related  to  spherical-like coordinates with a radial distance defined by $x^1$ and angles $x^2=\theta$ and $\varphi$. In the latter case  initial data should be periodic in $\varphi$ and satisfy appropriate conditions at $\theta=0,\pi$ (see section 3).

Let us choose the following basis of 1-forms and the dual vector basis 
\be\label{25}
\theta^a=dx^a\ ,\ \ \theta^3=d\varphi+\beta_a dx^a
\ee
\be\label{26}
e_a=\p_a-\beta_a\p_{\varphi}\ ,\ \ e_3=\p_{\varphi}
\ee
adapted to  metric (\ref{30}). It is convenient to define a complex coordinate $\xi$, a complex operator $\p$ and  functions $\beta$, $U$, $V$, $W$  as follows
\be\label{35}
\xi=x^1+ix^2\ ,\ \ \ \beta=\frac 12(\beta_1-i\beta_2)\ ,\ \ \partial=\p_{\xi}-\beta\p_{\varphi}
\ee
\be\label{36}
U=\frac {1}{2}\alpha (K_{11}-K_{22})-i\alpha K_{12}\ ,\ \ \ V=\alpha(K_{13}-iK_{23})\ ,\ \ W=K^{\ a}_a\ .
\ee
In coordinates $\xi,\bar\xi,\varphi$    metric  $g$ reads
\begin{equation}
g=\rho^2d\xi d\bar\xi+\alpha^2 (d\varphi+\beta d\xi+\bar\beta d\bar\xi)^2\ .\label{30f}
\end{equation}

\begin{prop}
 The momentum constraint (\ref{3a}) with  metric (\ref{30f}) decouples  into the following  system of equations for a real function $W$ and a complex function $U$
\ba\label{39b}
(\rho^{3}W)_{,\varphi}=E
\ea
\be
U_{,\bar\xi}-(\bar\beta U)_{,\varphi}=F\ ,\label{37}
\ee
where  $E$ is given by (\ref{39c}) and $F$ is given  up to $W$ by (\ref{39a}). Free data consist of metric $g$, complex function $V$, a real $\varphi$-independent function $W_0$ and a complex function $U_0$ satisfying 
\be
U_{0,\bar\xi}-(\bar\beta U_0)_{,\varphi}=0\ .\label{37a}
\ee
\end{prop}
\begin{proof}
Let $\hat K^{\ b}_a$ be the traceless part of $K^{\ b}_a$
\be
K^{\ b}_a=\hat K^{\ b}_a+\frac 12W\delta^{\ b}_a\ .\label{31}
\ee
In basis (\ref{26}) the momentum constraint (\ref{3a}) with $i=3$ yields (\ref{39b})
with
\ba\label{32a}
E=\rho\alpha^{-1}e_a(\alpha\rho^2 K_3^{\ a})-2\rho^3\beta_{a,\varphi}K_3^{\ a}\ .
\ea
For $i=a$ one obtains
\ba\label{33}
\rho^2(\alpha \hat K^{\ b}_a)_{|b}-(\alpha\rho^2\hat K^{\ b}_a\beta_b)_{,\varphi}=F_a\ ,
\ea
where ${}_{|b}$ denotes the covariant derivative with respect to metric $g_{ab}=\rho^2\delta_{ab}$ and
\ba\label{34}
F_a=-\frac 12\rho^2\alpha^{-2}e_a (\alpha^3W)+\frac 32\rho^2\alpha W\beta_{a,\varphi}+\rho^2\alpha(\eta^{cd}e_c\beta_d)\eta_{ab}K_3^{\ b}
-(\alpha^{-1}\rho^2K_{3a})_{,\varphi}.
\ea

By means of  (\ref{35}) and (\ref{36}) formula (\ref{32a}) can be written as
\ba\label{39c}
E=2\rho \alpha^{-1}Re(\bar\p V-2\bar\beta_{,\varphi}V)\ .
\ea
A complex combination of equation (\ref{33}) with $a=1$ and $a=2$ leads to (\ref{37})  with $F=\frac 12(F_1-iF_2)$ of the form
\ba
F=-\frac 12\rho^2\alpha^{-2}\p (\alpha^3 W)+\frac 32\rho^2\alpha W\beta_{,\varphi}-2i Im(\bar\p\beta)V
-\frac 12(\alpha^{-2}\rho^2V)_{,\varphi}.\label{39a}
\ea

 Equation (\ref{39b}) determines $W$ up to a real function $f(\xi,\bar\xi)$ 
\ba\label{38a}
W=\rho^{-3}(\int_{\varphi_0}^{\varphi}{Ed\varphi'}+f(\xi,\bar\xi))\ .
\ea 
Function $f$ is in one to one  correspondence to function  $W_0=W|_{\varphi=\varphi_0}$. 
Substituting (\ref{38a}) into (\ref{39a}) yields $F$ up to $f$. Then (\ref{37}) becomes an equation for $U$. Its solution, if it exists, is defined up to a solution $U_0$ of the homogenoeus  part of (\ref{37}).  

\end{proof}
Note that equations (\ref{39b})-(\ref{37}) are invariant under transformation (\ref{1a}) with $\tilde H=0$ since (\ref{3a}) is invariant. Using this freedom  one can fix one of components of $g$. Another component  can be fixed by transformation of coordinate $\varphi$.

Substituting $F=\hat F_{,\varphi}$ into (\ref{37}) leads to the equation
\be\label{38}
\bar\p \hat U=\hat F
\ee
for a function  $\hat U$ such that $U=\hat U_{,\varphi}$. Operator $\p$ defines the Cauchy-Riemann (CR) structure  (see \cite{bo} and references therein) on the initial manifold, not unique since locally there are many systems of coordinates in which metric takes the form (\ref{30}).
 The example of Hans Lewy \cite{l} shows that equation (\ref{38}) can be unsolvable for some functions $\beta$ and $\hat F$. CR structures are known to  appear in general relativity, especially in the context of algebraically special solutions of the Einstein equations (see e.g. \cite{lnt}).

There are several cases in which solutions of (\ref{38}) can be represented in an integral form. In all of them 
the CR structure is realizable. This means that, in addition to $\xi$, there exists a solution $\chi$ of equation $\bar\p\chi=0$ such that $\chi_{,\varphi}\neq 0$.
  Equivalently,  $\beta$ can be written  in the  form
\be
\beta=\frac{\bar\chi_{,\xi}}{\bar\chi_{,\varphi}}\ .\label{30a}
\ee
Given $\chi$ the initial manifold can be considered as a 3-dimensional real surface in space $C^2$ of pairs $(\xi,\chi)$. 

\begin{prop}
 If $\beta=0$ then $U$ is given by
 \be\label{45}
U=\frac{1}{2\pi i}\int_{\Omega}{\frac{F(\xi',\bar\xi',\varphi)}{\xi'-\xi}d\xi'\wedge d\bar\xi'}+h(\xi,\varphi)\ ,
\ee
where the integral is taken over a bounded open neighborhood ${\Omega}$ of $\xi$ in $C$ and $h$ is an arbitrary function holomorphic in $\xi$.
 \end{prop}
\begin{proof}
For $\beta=0$ equation (\ref{37}) reads
\be
U_{,\bar\xi}=F\ .\label{37x}
\ee
Integrating (\ref{37x}) with the fundamental solution $(\pi\xi)^{-1}$ for the Cauchy-Riemann operator $\p_{\bar\xi}$ leads to a version of Cauchy's integral formula  (Theorem 1.2.1 in \cite{h}) which can be converted into (\ref{45}). The domain of integration can be extended to a whole support of $F$ provided that  the r. h. s. of (\ref{45}) still  makes sense.

\end{proof}
\begin{prop}
 Let functions $\rho$, $\alpha$, $\chi$ and $V$ be analytic with respect to coordinate $\varphi$. Then
 \be\label{45c}
U=\chi_{,\varphi}\big( \frac{1}{2\pi i}\int_{\Omega}{\frac{F(\xi',\bar\xi',\varphi(\xi',\bar\xi',\chi))}{\xi'-\xi}\varphi_{,\chi}(\xi',\bar\xi',\chi)d\xi'\wedge d\bar\xi'}+h(\xi,\chi)\big)\ ,
\ee
where $\chi$  depends on unprimed coordinates and $\varphi (\xi',\bar\xi',\cdot)$ is the inverse function to $\chi(\xi',\bar\xi',\cdot)$.
\end{prop}
\begin{proof}
Under assumptions of this proposition functions $W$, $F$ and $\hat F$ are also analytic in $\varphi$. Thus, we can complexify $\varphi$ and pass to coordinates $\xi$, $\bar\xi$ and $\chi$. In these coordinates equation (\ref{38}) reads
\be\label{38b}
 \hat U_{,\bar\xi}=\hat F(\xi,\bar\xi,\varphi(\xi,\bar\xi,\chi))\ .
\ee
Now, we can follow  (\ref{45}) in order to represent $\hat U$ as 
\be\label{45d}
\hat U=\frac{1}{2\pi i}\int_{\Omega}{\frac{\hat F(\xi',\bar\xi',\varphi(\xi',\bar\xi',\chi))}{\xi'-\xi}d\xi'\wedge d\bar\xi'}+\hat h(\xi,\chi)\ .
\ee
If we come back to coordinates $\xi$, $\bar\xi$, $\varphi$ and differentiate (\ref{45d}) over $\varphi$ we obtain (\ref{45c}) with $h=\hat h_{,\chi}$.

\end{proof}
Formula (\ref{45c}) becomes much simpler if $\beta_{,\varphi}=0$. 
Then we can choose
\be
\chi=\varphi+\chi^0(\xi,\bar\xi)\ ,\label{41a}
\ee
where $\chi^0_{,\bar\xi}=\bar\beta$ (this condition can be always locally solved with respect to $\chi_0$). Expression (\ref{45c}) takes the form
\be\label{45f}
U=\frac{1}{2\pi i}\int_{\Omega}{\frac{F(\xi',\bar\xi',\varphi+\chi^0(\xi,\bar\xi)-\chi^0(\xi',\bar\xi'))}{\xi'-\xi}d\xi'\wedge d\bar\xi'}+h(\xi,\varphi+\chi^0(\xi,\bar\xi))\ .
\ee

If all initial data are independent of  $\varphi$ formula  (\ref{45f})  reduces to
\be\label{45a}
U=\frac{1}{2\pi i}\int_{\Omega}{\frac{F(\xi',\bar\xi')}{\xi'-\xi}d\xi'\wedge d\bar\xi'}+h(\xi)\ .
\ee
In this case equation (\ref{39b}) is replaced by condition $E=0$. It implies 
\begin{equation}
K^{\ a}_3=\alpha^{-1}\eta^{ab}\omega _{,b}\ ,\label{50}
\end{equation}
where $\omega$ is a function of coordinates $x^a$ \cite{tj} (a potential equivalent to $\omega$ was also introduced in \cite{bp,d,cm}). 
Now $F$ is determined by metric $g$ and functions $W$ and $\omega$. 

If $\chi$  and $\hat F$ are analytic in all coordinates then one can solve (\ref{38})   by integrating (\ref{38b}) over $\bar\xi$. Then, instead of (\ref{45c}) one obtains
\be\label{41}
U=\chi_{,\varphi}\big(\int_{\bar\xi_0}^{\bar\xi}{F(\xi,\bar\xi',\varphi(\xi,\bar\xi',\chi)\varphi_{,\chi}(\xi,\bar\xi',\chi)d\bar\xi'}+h(\xi,\chi)\big)\ .
\ee
This formula becomes particularly simple if data is independent of $\varphi$
\be\label{44a}
U=\int_{\bar\xi_0}^{\bar\xi}{F(\xi,\bar\xi')d\bar\xi'}+h(\xi)\ .
\ee

All considerations in this section are purely local. Assume that we can extend local solutions of the momentum constraint (\ref{3a}) to a whole initial manifold, e. g. to an asymptotically Euclidean manifold. In this case  the mean curvature $\tilde H$ must vanish. In order to complete the construction of initial data one should  solve the Lichnerowicz equation (\ref{2a}) for the conformal factor $\psi$.  Existence of a solution   is equivalent to positivity of the Yamabe type invariant \cite{ma,ma1}, but there is no practical device how to satisfy the latter condition. For some classes of data one can get it from the Sobolev inequalities admitted by the initial manifold \cite{tj}. 
The simplest way to assure existence and uniqueness of $\psi$ is to assume that  the initial manifold is complete and $R\geq 0$. For instance, one can  take the same initial  manifold and metric $g$ as in a known maximal ($\tilde H=0$) asymptotically flat solution  of the full set of initial conditions. Then $R\geq 0$ since the Hamiltonian constraint is satisfied by this solution. If we  take as  $K$  another solution of (\ref{3a}) with the same $g$  we can be sure that  the Lichnerowicz  equation can be solved   but there is still  problem to find new $\psi$ numerically.

We are not able to  extend  results of this section to data (\ref{1a}) with $\tilde H\neq const$. Then equation (\ref{3a}) is no longer equivalent to (\ref{1}). Equations (\ref{39b}) and  (\ref{37}) are replaced by 
\ba\label{50a}
(\rho^3W)_{,\varphi}=E-\frac 23\psi^6\rho^3 \tilde H_{,\varphi}
\ea
\be
U_{,\xi}-(\beta U)_{,\varphi}=F+\frac 23\psi^6\rho^2\alpha \bar\p \tilde H\label{50b}
\ee
( note that  $\beta$, $U$, $V$ and $W$ are preserved by  transformation (\ref{1a}) and $\rho\rightarrow\psi^2\rho$, $\alpha\rightarrow\psi^2\alpha$).
Equations (\ref{50a}) and (\ref{50b}) have to be considered simultaneously with (\ref{2a})  except the case $\tilde H_{,\varphi}=0$ for which $W$ is given by (\ref{38a}). Existence of solutions of this system is much more difficult to prove \cite{i,cbiy}.

\section{Horizons}
In the theory of black holes it is important to construct initial data with one or more 2-dimensional surfaces which can represent horizons of black holes. This can be done within the conformal method  by imposing an appropriate condition on the conformal factor $\psi$  on an internal boundary of the initial manifold \cite{ma}. This boundary becomes MOTS  for initial data $(\tilde g,\tilde K)$ given by  (\ref{1a}). Existence of $\psi$ for asymptotically flat data depends  on the positivity of the Yamabe type invariant what is even more difficult to prove than in the case of an unbounded initial manifold (see \cite{tj} for partial results). Another problem is that a continuation of initial data through this internal boundary is not assured. This problem does not appear in the inversion symmetry approach of Misner \cite{m} which is applicable to a  restricted class of data. In this section we  propose a construction of data which follows the latter method.

We would like to  generalize the Kerr metric data at $t=const$, where $t$ is the Boyer-Lindquist time coordinate. The  initial metric induced by the Kerr solution reads 
\be 
g=\rho^2\Delta^{-1}dr^2+\rho^2d\theta^2+\rho^{-2}\Sigma^2\sin^2{\theta}d\varphi^2\ ,\label{48}
\ee
where
\be
\rho^2=r^2+a^2\cos^2{\theta}\ ,\ \ \Delta=r^2-2Mr+a^2\ ,\ \Sigma^2=(r^2+a^2)^2-a^2\Delta\sin^2{\theta}\ .\label{49}
\ee
Metric (\ref{48}) takes the form (\ref{30}) in coordinates $x^i=\tilde r,\theta,\varphi$, where 
 $\tilde r$  is related to $r$ by
\be\label{h24a}
r=M+\sqrt{M^2-a^2}\cosh{\tilde r}\ ,\ \tilde r\in [-\infty,\infty]\ .
\ee
 Domains where $\tilde r>0$ or $\tilde r<0$ are asymptotically flat. They are connected by the external Kerr horizon (the Einstein-Rosen bridge)   located at $\tilde r=0$.
Only non-vanishing components of $K$  are given  by (\ref{50}) with  \cite{tj}
\be
\omega=4aM\rho^{-2}[2(r^2+a^2)+(r^2-a^2)\sin^2{\theta}]\cos{\theta}\ .\label{51}
\ee
The Kerr initial data  are analytic in coordinates $x^i$. Tensors $g$ and $K$ are invariant  under the reflection 
\be
x^1\rightarrow -x^1\label{51a}
\ee
which corresponds to the inversion invariance of Misner. This and vanishing of $K^{11}$ at $x^1=0$ imply that surface $x^1=0$ is MOTS. 

In order to generalize the Kerr data we assume that initial data are given by (\ref{1a}), metric $g$ has the form (\ref{30f}) and
\be\label{51n}
g\rightarrow g
\ee
\be\label{51d}
K\rightarrow \epsilon K\ ,\ \epsilon=\pm 1  
\ee
under reflection (\ref{51a}).
 Transformations (\ref{51n})-(\ref{51d}) are equivalent to the following transformations of components of $g$ and $K$
 \be
\alpha\rightarrow \alpha\ ,\ \rho\rightarrow \rho\ ,\ \beta \rightarrow -\bar\beta\label{51g}
\ee
\be
V\rightarrow -\epsilon \bar V\ ,\ \ W\rightarrow \epsilon W\ ,\ \ U\rightarrow \epsilon \bar U\ .\label{51h}
\ee
\begin{prop}
(i) Equations (\ref{39b})-(\ref{37})  are compatible with transformations (\ref{51a}) and (\ref{51g})-(\ref{51h}). The same is true for formulas  (\ref{38a}), (\ref{45}), (\ref{45c}),  (\ref{45f})-(\ref{44a}) provided that 
\be\label{51e}
\chi\rightarrow \bar \chi\ ,\ f\rightarrow \epsilon f\ ,\ h\rightarrow \epsilon\bar h\ ,\ \omega\rightarrow -\epsilon\omega\ .
\ee
(ii) Let $g$ and $K$ satisfy (\ref{39b})-(\ref{37}), (\ref{51g})-(\ref{51h}) and condition (trivial for  $\epsilon=-1$)
 \be\label{70}
 ReU=-\frac 12\alpha\rho^2 W\ \ at\ \ x^1=0\ .
 \ee
If the Lichnerowicz equation (\ref{2a}) admits a unique solution  $\psi$ then surface $x^1=0$ is MOTS with respect to  initial data $(\tilde g,\tilde K)$  given by (\ref{1a}) with $\tilde H=0$.  
 \end{prop}
 \begin{proof} The part (i) of the proposition can be easily proved by considering induced transformations of functions $E$ and $F$ given by (\ref{39c}) and  (\ref{39a}).  The same refers to (\ref{50}) and integral formulas for $W$ and $U$  in section 2 if transformations (\ref{51e}) are taken into account.

  The reflection invariance of $g$  implies that the exterior curvature of  surface $x^1=0$ embedded in the initial manifold  vanishes. For $\epsilon=-1$ function $K^{11}$ is antisymmetric with respect to $x^1$, hence automatically
  \be\label{51b}
K^{ij}n_in_j=0\ \ at\ \ x^1=0\ ,
\ee
where $n_i$ is the normal vector of the surface. For $\epsilon=1$ equality (\ref{51b}) is assured by (\ref{70}). The Lichnerowicz equation is invariant under the transformation given by (\ref{51a})-(\ref{51d}). If its solution $\psi$ is unique it must be also invariant. Then properties (\ref{51d}) and (\ref{51b}) are preserved by the conformal  transformation (\ref{1a}). Hence, the exterior curvature of  surface $x^1=0$ with respect to metric $
\tilde g$ also vanishes. This together with  equation (\ref{51b}) and $\tilde H=0$  imply zeroing of expansions of null rays emitted outward or inward from $x^1=0$. Hence, this surface is  MOTS with respect to data $(g,K)$ and also to data $(\tilde g,\tilde K)$.
 
 \end{proof}
 
 We would like a MOTS to be diffeomorphic to the sphere $S_2$. Following the Kerr data let us assume that $x^1$ plays a role of a radial coordinate and surfaces $x^1=const$ are spheres with the angular coordinates $x^2=\theta$ and $\varphi$. Initial data  should be periodic in $\varphi$ and regular at $\theta=0,\pi$. Let $\rho$, $\alpha$, $\beta_a$ and $K_{3a}$ be invariant under the translation $\varphi\rightarrow\varphi+2\pi$. Then functions $E$ and $F$ in equations  (\ref{39b}) and (\ref{37}) are also invariant. Formula (\ref{38a}) defines periodic $W$ provided that
 \be 
\int_0^{2\pi}{Ed\varphi}=0\ .\label{82}
 \ee
 We can achieve (\ref{82}) in different ways e. g. by shifting one of the variables $\beta_a$, $K_{3a}$ by a $\varphi$-independent function. Alternatively, one can use equation (\ref{39b}) in order to define $\rho$ (or $\alpha$) in terms of periodic functions $\alpha$ (or $\rho$), $\beta_a$, $\alpha K_{3a}$ and $\rho^3W$. In this case we have to care about positivity of the resulting expression. In order to assure   periodicity of  $U$ equation (\ref{37}) should be completed by the condition
 \be
 U(\varphi+2\pi)-U(\varphi)=0 \label{83}
 \ee
 which is compatible with  (\ref{37}).
 Note that all formulas  (\ref{45}), (\ref{45c}), (\ref{45f}), (\ref{41}) define periodic $U$ under suitable assumptions on functions $\chi$ and $h$.
 
 In order to avoid sigularities of metric (\ref{30}) at $\theta=0,\pi$ we assume  that $\rho\neq 0$ everywhere, $\alpha\neq 0$ if $\theta\neq 0,\pi$ and 
 \be
 \alpha\approx\rho\sin{\theta},\ \ \beta_2\approx 0\label{84}
 \ee
 near $\theta=0,\pi$.
 Forms $\sin{\theta}d\theta$, $\sin^2{\theta}d\varphi$ and tensor $d\theta^2+\sin^2{\theta}d\varphi^2$ are nonsingular on $S_2$. In order to assure  regularity of $K$  we assume that near $\theta=0,\pi$  components  $K_{ij}$ of $K$ are proportional to the following powers of $\sin{\theta}$
 \be
 K_{ij}\sim(\sin{\theta})^{i+j-2}\label{85}
 \ee
 or 
 \be\label{86}
 K_{33}\approx K_{22}\sin^2{\theta}
 \ee
 and remaining functions $K_{ij}$ satisfy (\ref{85}). It follows from these conditions that
 \be\label{87}
  ReV\sim\sin^3{\theta}\ ,\ \ ImV\sim\sin^4{\theta},\ \ E\sim\sin^2{\theta}\ ,\ \ ReF\sim\sin{\theta}
  \ee
  near $\theta=0,\pi$
  and  $U$  should behave as
  \be
  ReU\sim\sin{\theta}\ ,\ \ ImU\sim\sin^2{\theta}\ .
 \ee
 Above conditions are compatible with equations (\ref{39b})-(\ref{37}). 
 
 MOTS has a chance to develop into the event horizon if initial data are asymptotically Euclidean. Let $r=ce^{x^1}$, where $c=const$,  approximates  the radial distance  in a flat space if $x^1\rightarrow\infty$.  Standard asymptotical  conditions (see e. g. \cite{t}) are equivalent to 
 \be
 \rho=r+0_2(r^{1-\epsilon})\ ,\ \ \alpha=r\sin{\theta}+0_2(r^{1-\epsilon})\ ,\ \  \beta_a=0_2(r^{-\epsilon})\ , \
K_{ij}=0_1(r^{1-\epsilon})\ .\label{88}
\ee
Here $\epsilon$ is a positive constant and we write $f=0_k(r^l)$ if derivatives of $f$  of order $n\leq k$ fall off as $r^{l-n}$ when $r\rightarrow \infty$.
 
 According to proposition 3.1  one can relatively easily construct initial data with   MOTS assuming  symmetry (\ref{51d}) with $\epsilon=-1$. For instance, let $g$  be the Kerr initial metric (\ref{48}) in coordinates $\tilde r$, $\theta$, $\varphi$ and   only non-vanishing components of  $K$ are given by (\ref{50}) with $\omega$ which is an even function of $x^1$. Since $R\geq 0$ the Lichnerowicz equation with $\tilde H=0$ admits a unique solution $\psi$ which is reflection symmetric. Given $\psi$ transformation (\ref{1a}) yields ultimate  data which satisfies all initial constraints. Note that we can choose $\omega$ which tends to (\ref{51}) if $x^1=\tilde r\rightarrow \infty$. Then the data approaches the Kerr data with the angular momentum $a$ if $x^1\rightarrow \infty$ and with the angular momentum $-a$ if $x^1\rightarrow -\infty$.

 For $\epsilon=1$  condition (\ref{70}) imposes an inconvenient constraint on integral representations of $W$ and $U$ from section 2. Function $f$ in (\ref{38a}) cannot solve this constraint because $f$ is independent of $\varphi$. Function $h$ in integral formulas for $U$ can be used only for analytic fields. 
 A radical way to avoid  problem with (\ref{70}) is to assume  $W=ReU=\beta=0$. Then equation (\ref{39b}) yields $V=2i\omega_{,\xi}$ and (\ref{37}) is equivalent to the following  conjugate equations for $ImU$
 \ba\label{77}
 (ImU)_{,\bar\xi}=-(\rho^2\alpha^{-2}\omega_{,\xi})_{,\varphi}\ ,\ \ (ImU)_{,\xi}=-(\rho^2\alpha^{-2}\omega_{,\bar\xi})_{,\varphi}\ .
 \ea
The integrability condition for this system reads
\be\label{78}
(\rho^2\alpha^{-2}\omega_{,1})_{,2}+(\rho^2\alpha^{-2}\omega_{,2})_{,1}=f_{,2}\ ,
\ee
where derivatives are taken with respect to real coordinates $x^a$ and $f$ is a $\varphi$-independent  function. One can consider (\ref{78}) as an equation for $\omega$, in which $\varphi$ is a parameter, or one can formally solve (\ref{78}) to obtain
\be\label{79}
\rho^2\alpha^{-2}\omega_{,1}-f=\gamma_{,1}\ ,\ \ \rho^2\alpha^{-2}\omega_{,1}=-\gamma_{,2}\ ,
\ee
where $\gamma$ is a new function. Let us treat (\ref{79}) as equations for $\omega$ and $\rho^2\alpha^{-2}$ whereas $\gamma$ and $f$ are given. Then $\omega$ has to satisfy the first order linear equation
\be\label{80}
\gamma_{,2}\omega_{,1}+(\gamma_{,1}+f)\omega_{,2}=0\ .
\ee
If it is solved function $\rho^2\alpha^{-2}$ is given by 
\be\label{81}
\frac{\rho^2}{\alpha^{2}}=\frac{(\gamma_{,1}+f)\omega_{,1}-\gamma_{,2}\omega_{,2}}{\omega_{,1}\omega_{,1}+\omega_{,2}\omega_{,2}}
\ee
provided that the nominator and the denominator of the r. h. s. of (\ref{81}) are both positive. Solving (\ref{80}) with respect to $\omega$ is equivalent to finding integral lines of the vector field $v=\gamma_{,2}\p_1+(\gamma_{,1}+f)\p_2$ in $R^2$. The latter problem can be reduced to an ordinary differential equation with an arbitrary initial condition.
 \section{Summary}
We have been considering  initial constraints for the vacuum Einstein equations in the framework of the conformal approach.  We assumed that  metric can be put into  form (\ref{30f}). One can relate with  coordinates $\xi,\bar\xi,\varphi$  the Cauchy-Riemann operator $\p$.   The momentum constraint with $\tilde H=const$ splits into equations  (\ref{39b}) and (\ref{37}) (proposition 2.1). Equation (\ref{39b}) can be directly integrated with respect to $W$ giving formula (\ref{38a}). In the case of fields analytic in coordinate $\varphi$ or data with $\beta=0$  solutions of equation (\ref{37})  can be also written as    integrals of  known functions (propositions 2.2 and 2.3). In order to complete the construction of initial  data one has still  to solve  the Lichnerowicz equation (\ref{2a}) for the conformal factor $\psi$. Its existence and uniqueness  can be  easily proved in some cases, e.g. if data $(g,K)$ are asymptotically flat, $\tilde H=0$ and the 
Ricci scalar of $g$ is nonnegative.

In section 3 we propose a construction of maximal data with a reflection symmetry which, together with (\ref{70}), implies  existence of a horizon in the form of a marginally outer trapped surface (proposition 3.1).  As an example we present data  obtained by a modification of  the Kerr initial data.

\end{document}